\documentclass[letterpaper,journal,twoside,final]{IEEEtran}

\usepackage{url}
\usepackage{epsf}
\usepackage{graphicx}
\usepackage{psfrag}
\usepackage{amsmath,amssymb,amsfonts}
\usepackage{color}
\usepackage{subfigure}
\usepackage{fancyhdr}
\usepackage{verbatim,moreverb}
\usepackage{pifont}
\usepackage{setspace,supertabular}
\usepackage{mathrsfs}
\usepackage{bbm}
\newtheorem{thm}{Theorem}

\newtheorem{lem}[thm]{Lemma}
\newtheorem{prop}[thm]{Proposition}

\newtheorem{defn}[thm]{Definition}

\def\P{{\mathsf P}}

\newcommand{\Ti}{\mathcal{T}}
\newcommand{\uh}{\hat{s}}

\newcommand{\aep}{\Ti_{\epsilon}^{(n)}}

\newcommand{\e}{\epsilon}

\begin{document}

\title{Correlated Sources over Broadcast Channels}

\author{Paolo Minero and Young-Han Kim%
\thanks{The authors are with the Information Theory and Applications Center (ITA) of the California Institute of Telecommunications and Information Technologies (CALIT2), Department of Electrical and Computer Engineering, University of California, San Diego CA, 92093, USA. Email: pminero@ucsd.edu, yhk@ucsd.edu.}
\thanks{This work is partially supported by the National Science Foundation
award CNS-0546235 and CAREER award CCF-0747111.}
}

\maketitle

\begin{abstract}
The problem of reliable transmission of correlated sources over the
broadcast channel, originally studied by Han and Costa, is revisited.
An alternative characterization of their sufficient condition for
reliable transmission is given, which includes results of Marton for
channel coding over broadcast channels and of Gray and Wyner for
distributed source coding. A ``minimalistic'' coding scheme is
presented, which is based on joint typicality encoding and decoding,
without requiring the use of Cover's superposition coding, random
hashing, and common part between two sources. The analysis of the
coding scheme is also conceptually simple and relies on a new
multivariate covering lemma and an application of the Fourier--Motzkin
elimination procedure.
\end{abstract}

\section{Introduction}
We study the problem of broadcasting two arbitrarily correlated
sources over a general two-receiver discrete memoryless broadcast
channel (DM-BC) as depicted in Figure \ref{fig:1}. A discrete
memoryless stationary source produces $n$ independent copies $(S_1^n,
S_2^n) \triangleq \{(S_{1,i},S_{2,i})\}_{i=1}^n$ of a pair of random
variables $(S_1,S_2) \sim p(s_1,s_2)$.  The encoder maps the source
sequences $(S_1^n, S_2^n)$ into a sequence $X^n$ from a finite set
$\cal X$ and broadcasts it to two separate receivers. The
communication channel is memoryless and modeled by a transition
probability matrix $p(y_1,y_2|x)$, which maps each channel input $x$
into a pair of finite valued output symbols $(y_1,y_2)$.  Decoder 1
maps the channel output sequence $Y^n_1$ into an estimate
$\hat{S}_1^n$ of the source sequence $S^n_1$. Similarly, decoder 2
maps $Y^n_2$ into an estimate $\hat{S}_2^n$ of the source sequence
$S_2^n$. For a given encoder and decoders, the average probability of
error is defined as
\begin{equation*}
\text{Pr}\{(\hat{S}^n_1,\hat{S}^n_2) \neq  (S^n_1, S^n_2) \}.
\end{equation*}
We say that the sources $(S_1,S_2)$ can be reliably transmitted over
the DM-BC $p(y_1,y_2|x)$ if there exists a sequence of encoders and
decoder pairs such that the average probability of error vanishes as
$n \to \infty$.

\begin{figure}
\begin{center}
\scalebox{0.89}{
\psfrag{W}[b]{\hspace{0.9em} $(S_1^n,S_2^n)$}
\psfrag{W1}[b]{\hspace{-0.5em} $\hat{S}_1^n$}
\psfrag{W2}[b]{\hspace{-0.5em} $\hat{S}_2^n$}
\psfrag{A}[cl]{\hspace{-1.7em} {\small Encoder}}
\psfrag{X}{\hspace{-0.9em} $X^n$}
\psfrag{B}[cl]{\hspace{-2.4em} {\small Decoder 1}}
\psfrag{C}[cl]{\raisebox{0.2em}{\hspace{-2.2em} {\small Decoder 2}}}
\psfrag{Y1}{\hspace{-0.9em} $Y_1^n$}
\psfrag{Y2}{\hspace{-0.9em} $Y_2^n$}
\psfrag{p}{\hspace{-2.39em} $p(y_1,y_2|x)$}
\includegraphics[width=3.6in]{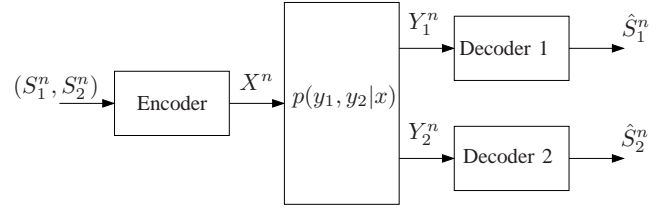}}
\end{center}
\caption{The two-receiver broadcast channel with arbitrarily correlated sources.} \label{fig:1}
\end{figure}

Han and Costa~\cite{HanCosta} provide a sufficient condition for
reliable transmission of the sources $(S_1,S_2)$ over a general DM-BC
$p(y_1,y_2|x)$. Necessary conditions are derived by Gohari and
Anantharam~\cite{Gohari}, and Kramer \textit{et
  al.}~\cite{KramerLiangShamai}. Finding matching sufficient and
necessary conditions is still an open problem in general.

Two special classes of sources and channels have been studied
extensively in the literature.  On the one hand, suppose
$S_1=(W_0,W_1)$ and $S_2=(W_0,W_2)$, where $W_0$, $W_1$, and $W_2$ are
three mutually independent random variables of entropies $R_0$, $R_1$,
and $R_2$, respectively. The (closure of the) set of triples
$(R_0,R_1,R_2)$ which can be reliably transmitted is called the {\em
  capacity region} of the DM-BC. Inner bounds on the capacity region
of a general DM-BC are developed in \cite{Cover, Csiszar, Gelfand,
  Liang, Marton, vanderMeulen}, and the most general known achievable
rate region is due to Marton~\cite{Marton}, which is tight for several
classes of broadcast channels.  On the other hand, suppose the DM-BC
is noiseless with $Y_1 = (X_0,X_1)$, $Y_2 = (X_0,X_2)$, where $X_0$,
$X_1$, and $X_2$ have cardinalities $2^{R_0}$, $2^{R_1}$, and $2^{R_2}$,
respectively. The (closure of the) set of triples $(R_0,R_1,R_2)$
which can be reliably transmitted is called the {\em distributed
  source coding region}, and its complete characterization is given by
Gray and Wyner~\cite{Gray}.

For transmission of arbitrarily correlated sources over a general
noisy DM-BC, conclusive results are known for the case in which either
source or channel has a degraded structure.  More specifically, Han
and Costa's sufficient condition matches known necessary conditions if
decoder 1 is interested in recovering both source sequences
$(S_1^n,S_2^n)$ while decoder 2 is interested in $S_2^n$ only
(degraded source sets)~\cite{KangKramer} or if the DM-BC is degraded,
or more generally, ``more capable''~\cite{ElGamal2,
  KramerLiangShamai}.

In this paper, we provide an alternative characterization to the
sufficient condition of Han and Costa for reliable transmission of
correlated sources. This new characterization, while algebraically
equivalent to Han and Costa's original sufficient condition, is
simpler and includes in a straightforward manner Marton's inner bound
on the broadcast channel capacity region, Gray--Wyner distributed
source coding region, and other aforementioned known results. In
particular, the new characterization does not involve the common part
of the two sources in the sense of G\'{a}cs, K\"orner, and
Witsenhausen~\cite{gaks,Witsenhausen}. This implies that there is no special role
played by the common part of the sources, confirming the standard
engineering intuition.  It is interesting to observe the difference to
the problem of reliable transmission of correlated sources over
{\em multiple access channels}. For this problem, the common part of the
sources plays a pivotal role in Cover, El Gamal, and Salehi's random
coding construction~\cite{CoverSalehi} by inducing coherent
transmission from separate encoders.

The highlight of the paper is a simple coding scheme for the
sufficient condition and its analysis. This coding scheme is based on
joint typicality encoding and decoding, and does not involve random
hashing (Slepian--Wolf binning), Cover's superposition coding, or rate
splitting as in the original coding scheme by Han and Costa. The
performance of the coding scheme is analyzed via a new multivariate
covering lemma and an application of the Fourier--Motzkin elimination
procedure.

The rest of the paper is organized as follows. In the next section, we
review Han and Costa's sufficient condition (with a recent correction
by Kramer and Nair~\cite{KramerNair}). Section~III presents the
alternative characterization of the sufficient condition. The proof of
this new coding theorem is given in Section~IV.  We conclude the paper
in Section V with a discussion on implications of our coding scheme,
and a sketch of yet another coding scheme for broadcasting correlated
sources. Throughout the paper, we use the notation in \cite{ElGamalKim}.

\section{Han and Costa's Coding Theorem}
For convenience, we recall here the coding theorem of Han and Costa
with a
few minor changes.
\begin{thm}[\cite{HanCosta}, \cite{KramerNair}]
\label{thmHC}
The pair of sources $(S_1, S_2)$ can be reliably transmitted over
the DM-BC $p(y_1,y_2|x)$ if
\begin{align}
H(S_1) &< I(U_0, U_1, S_1; Y_1) - I(U_0, U_1; S_2|S_1), \nonumber \\
H(S_2) &< I(U_0, U_2, S_2; Y_2) - I(U_0, U_2; S_1|S_2), \nonumber \\
H(S_1,S_2) &< I(K, U_0, U_1, S_1; Y_1) + I(U_2, S_2; Y_2|K, U_0) \nonumber \\
&  \quad               - I(U_1, S_1; U_2,S_2|K,U_0),  \label{HC1} \\
H(S_1,S_2) &< I(U_1, S_1;Y_1|K, U_0) + I(K, U_0, U_2, S_2; Y_2) \nonumber \\
&  \quad               - I(U_1, S_1; U_2, S_2|K, U_0), \label{HC2}  \\
H(S_1,S_2) &< I(U_0, U_1, S_1; Y_1) + I(U_0, U_2, S_2; Y_2) \nonumber \\
&  \quad      - I(U_1, S_1; U_2, S_2|K, U_0) - I(S_1, S_2; K, U_0), \label{HC3}
\end{align}
for some $p(u_0,u_1,u_2,x|s_1,s_2)$. Here $K=f(S_1)=g(S_2)$ denotes
the common variable in the sense of G\'{a}cs, K\"orner, and
Witsenhausen, and the auxiliary random variable $U_0$ has the
cardinality bound $|{\cal U}_0| \le \min \{ |{\cal X}| |{\cal S}_1|
|{\cal S}_2| + 4, |{\cal Y}_1| |{\cal Y}_2| |{\cal S}_1| |{\cal S}_2|
+ 4 \}$.
\end{thm}

The coding scheme by Han and Costa requires three auxiliary random
variables $(U_0,U_1,U_2)$. Roughly speaking, $U_0$ serves as a ``cloud
center'' distinguishable by both receivers and carries the common part
$K$ and hash indexes of $(S_1,S_2)$, $U_1$ and $U_2$ are codewords
within the cloud centers which encode the remaining uncertainty about
$S_1$ and $S_2$, respectively.  Finally, Marton's subcode generation
technique~\cite{Marton} is employed to obtain arbitrary correlation
among $(U_0,U_1,U_2)$. For details of this interpretation, see
\cite{HanCosta}.

\section{An Alternative Characterization}
\label{sec3}
We have the following:
\begin{thm}
\label{thm1}
The sources $(S_1, S_2)$ can be reliably transmitted over the DM-BC
$p(y_1,y_2|x)$ if
\allowdisplaybreaks
\begin{align}
H(S_1) &< I(U_0, U_1, S_1; Y_1) - I(U_0, U_1; S_2|S_1),  \nonumber \\
H(S_2) &< I(U_0, U_2, S_2; Y_2) - I(U_0, U_2; S_1|S_2), \nonumber \\
H(S_1, S_2) &< I(U_0, U_1, S_1; Y_1) + I(U_2, S_2; Y_2|U_0) \nonumber \\
&  \quad               - I(U_1, S_1; U_2, S_2|U_0), \label{KM1} \\
H(S_1, S_2) &< I(U_1, S_1;Y_1|U_0) + I(U_0, U_2, S_2; Y_2) \nonumber \\
&  \quad               - I(U_1, S_1; U_2, S_2|U_0),  \label{KM2} \\
H(S_1, S_2) &< I(U_0, U_1, S_1; Y_1) + I(U_0, U_2, S_2; Y_2) \nonumber \\
&  \quad      - I(U_1, S_1; U_2, S_2|U_0) - I(S_1, S_2; U_0). \label{KM3}
\end{align}
for some $p(u_0,u_1,u_2,x|s_1,s_2)$. Here the auxiliary random
variable $U_0$ has the cardinality bound $|\mathcal{U}_0| \le \min \{ |{\cal X}| |{\cal S}_1|
|{\cal S}_2| + 4, |{\cal Y}_1| |{\cal Y}_2| |{\cal S}_1| |{\cal S}_2| + 4\}$.
\end{thm}

\emph{Remark 1:} It is easy to see that Theorem~\ref{thmHC} is a
special case of Theorem~\ref{thm1}. In fact, if we define a new random
variable as $\tilde{U}_0=(U_0,K)$, where $K$ denotes the common part
of the sources $(S_1,S_2)$, then the inequalities in
Theorem~\ref{thm1} for the triple $(\tilde{U}_0,U_1,U_2)$ simplify and
reduce to those in Theorem~\ref{thmHC}. Conversely, the following
proposition proves that Theorem~\ref{thm1} is a special case of
Theorem \ref{thmHC}, establishing the equivalence of the two.
\begin{prop}
\label{prop}
The sufficient condition for reliable transmission in Theorem
\ref{thmHC} is equivalent to the one in Theorem~\ref{thm1}.
\end{prop}
\begin{proof}
See Appendix A.
\end{proof}

\emph{Remark 2:} No cardinality bounds are known for the auxiliary random variables $U_1$ and $U_2$.

\emph{Remark 3:}  Application of Theorem~\ref{thm1} yields the following results as special cases:
\begin{itemize}
  \item[a)] \textit{Marton's inner bound \cite{Marton}}: Consider the special case of independent sources described above, so take $S_1=(W_0,W_1)$ and $S_2=(W_0,W_2)$, with $W_0$, $W_1$, and $W_2$ of entropies $R_0$, $R_1$, $R_2$, respectively. By choosing $(U_0,U_1,U_2)$ to be independent of $(W_0,W_1,W_2)$, Theorem~\ref{thm1} yields Marton's inner bound, which state that the capacity region of the DM-BC $p(y_1,y_2|x)$ contains the set of rate triples $(R_0,R_1,R_2)$ satisfying
\begin{align*}
R_0 + R_1 &< I(U_0,U_1; Y_1), \\
R_0 + R_2 &< I(U_0,U_2; Y_2),\\
R_0 + R_1 + R_2 &< I(U_0,U_1; Y_1) + I(U_2; Y_2|U_0) \\
&  \quad - I(U_1;U_2|U_0),\\
R_0 + R_1 + R_2 &< I(U_1; Y_1|U_0) + I(U_0,U_2; Y_2) \\
&  \quad - I(U_1;U_2|U_0),\\
2R_0 + R_1 + R_2 &< I(U_0,U_1; Y_1) + I(U_0,U_2; Y_2) \\
&  \quad - I(U_1;U_2|U_0)
\end{align*}
for some $p(u_0, u_1, u_2,x)$.

\item[b)] \textit{The Gray--Wyner source coding problem~\cite{Gray}}:
  Consider the noiseless DM-BC channel with $Y_1 = (X_0,X_1)$, $Y_2 =
  (X_0,X_2)$ with links of rate $R_0=H(X_0)$, $R_1=H(X_1)$,
  $R_2=H(X_2)$. By taking $U_1 = X_1$, $U_2 = X_2$, and $U_0 =
  (X_0,V)$, under the distribution $p(v|s_1,s_2)p(x_0)p(x_1)p(x_2)$,
  the inequalities in Theorem~\ref{thm1} simplify to the following
  rate region
\begin{align*}
R_0 + R_1 &> I(S_1,S_2; V) + H(S_1|V), \\
R_0 + R_2 &> I(S_1,S_2; V) + H(S_2|V), \\
R_0 + R_1 + R_2 &> I(S_1,S_2; V) + H(S_1|V) + H(S_2|V), \\
2R_0 + R_1 + R_2 &> 2I(S_1,S_2; V) + H(S_1|V) \\
& \quad + H(S_2|V)
\end{align*}
which includes (and is in fact equivalent to) the Gray--Wyner distributed source coding region, which is characterized by the following set of inequalities:
\begin{align*}
R_0 &> I(S_1,S_2; V),  \\
R_1 &> H(S_1|V), \\
R_2 &> H(S_2|V).
\end{align*}

\item[c)] \textit{Degraded ``source'' sets}: Consider the case in
  which decoder 1 is interested in reconstructing both source
  sequences $(S_1^n,S_2^n)$ with vanishing error probability. This
  setup can be captured by considering transmission of a pair of
  sources $\tilde{S}_1=(S_1,S_2)$, $\tilde{S}_2=S_2$. By choosing
  $U_0=(U,S_2)$, $U_1=X$, and $U_2=\text{const.}$ with
  $p(u,x|s_1,s_2)=p(u,x)$, the conditions stated in the
  Theorem~\ref{thm1} reduce to
      \begin{equation}
\setlength\arraycolsep{0.2em}
  \begin{array}{rcl}
H(S_2) & \le & I(U; Y_2),  \\
H(S_1, S_2) & \le & I(U; Y_2) + I(X; Y_1|U), \\
H(S_1, S_2) & \le & I(X; Y_1).
\end{array}
\label{eq:k1}
\end{equation}
Conversely, \cite{KramerLiangShamai} shows that if the sources
$(S_1,S_2)$ can be reliably transmitted in this setup, then conditions
$\eqref{eq:k1}$ with ``$<$'' replaced by ``$\le$'' have to be
simultaneously satisfied for some $p(u,x)$.  Thus, we conclude that
the sufficient condition in Theorem \ref{thm1} is also necessary for
this class of problems. It should be remarked that the same
condition can also be obtained by separately performing
source and channel coding, first compressing the sources and then encoding the resulting sequences using
the channel coding scheme in~\cite{Korner}.
\item[d)] \textit{More capable broadcast channels}: A DM-BC is said to
  be more capable if $I(X;Y_1) \ge I(X;Y_2)$ for all $p(x)$
  (cf.~\cite{ElGamal2}). If we apply Theorem \ref{thm1} to this
  setting, and we choose $U_0=(U,S_2)$, $U_1=X$, and
  $U_2=\text{const.}$, with $p(u,x|s_1,s_2)=p(u,x)$, then the
  conditions stated in the theorem reduce to
  $\eqref{eq:k1}$. Conversely, it is shown by Kang and
  Kramer~\cite{KangKramer} that if the sources $(S_1,S_2)$ can be
  reliably transmitted in this setup, then conditions $\eqref{eq:k1}$
  with ``$<$'' replaced by ``$\le$'' have to be simultaneously
  satisfied for some $p(u,x)$. Hence, the sufficient condition in
  Theorem \ref{thm1} is also necessary for this class of problems, but
  again the same condition can be obtained by performing source and
  channel coding separately.
\end{itemize}

\section{Proof of Theorem \ref{thm1}}
The coding technique used to prove Theorem \ref{thm1} can be outlined as follows. The scheme requires three auxiliary random variables $(U_0,U_1,U_2)$. Information about the source sequence $S_1$ is carried by $U_0$ and $U_1$, while information about the source sequence $S_2$ is carried by $U_0$ and $U_2$. Although $U_0$ carries information about both source sequences, it is not treated as cloud center in our coding scheme, so an error in decoding $U_0$ does not preclude successful decoding of $S_1$ or $S_2$.  Marton's subcode generation technique is used to obtain arbitrary correlation among $(U_0,U_1,U_2)$.

The following definitions are needed for the remainder of
the paper.
\begin{defn}
Let $(x^n, y^n)$ be a pair of sequences with elements drawn from a pair of finite alphabets $(\mathcal{X},\mathcal{Y})$. Define their joint empirical probability mass function as
\begin{equation*}
\pi(x, y|x^n, y^n) \triangleq \frac{|\{i : (x_i, y_i) = (x, y)\}|}{n}, \text{ for } (x,y) \in \mathcal{X}\times \mathcal{Y}
\end{equation*}
Then, the set $\Ti_{\epsilon}^{(n)}(X,Y)$ (in short, $\Ti_{\epsilon}^{(n)}$) of jointly $\epsilon$-typical $n$-sequences $(x^n, y^n)$ is defined as:
\small
\begin{align*}
\Ti_{\epsilon}^{(n)}(X,Y) \triangleq \{(x^n, y^n): & |\pi(x, y|x^n, y^n) -p(x,y)| \le \epsilon p(x,y) \\
& \;  \text{ for all } (x,y) \in \mathcal{X}\times \mathcal{Y} \}.
\end{align*}
\normalsize
Similarly, the set $\Ti_{\epsilon}^{(n)}(Y|x^n) $ of $\epsilon$-typical $n$-sequences $y^n$ that are jointly typical with a given sequence $x^n$ is defined as:
\begin{align*}
\Ti_{\epsilon}^{(n)}(Y|x^n) \triangleq \{y^n:  (x^n, y^n) \in \Ti_{\epsilon}^{(n)}(X,Y)   \}.
\end{align*}
\end{defn}
In the remaining of this section, we first describe the random codebook generation and encoding-decoding scheme, then we outline the analysis of the probability of error, which is treated in detail in Appendix C.

\emph{Random codebook generation}:
Let $\epsilon^{\prime}>0$. Fix a joint distribution $p(u_0,u_1,u_2|s_1,s_2)$ and, without loss of generality, let $p(x|u_0,u_1,u_2,s_1,s_2)$ be a chosen deterministic function $x(s_1,s_2,u_0,u_1,u_2)$. Compute $p(u_0)$, $p(u_1|s_1)$, and $p(u_2|,s_2)$ for the given source distribution $p(s_1,s_2)$. Randomly and independently generate $2^{nR_0}$ sequences $u_0^n(m_0)$, $m_0 \in [1:2^{nR_0}]$, each according to $\prod_{i=1}^n p_{U_0}(u_{0i})$. For each source sequence $s_1^n$ randomly and independently generate $2^{nR_1}$ sequences $u_1^n(s_1^n,m_1)$, $m_1 \in [1:2^{nR_1}]$, according to $\prod_{i=1}^n p_{U_1|S_1}(u_{1i}|s_{1i})$. The same procedure, using $\prod_{i=1}^n p_{U_2|S_2}(u_{2i}|s_{2i})$, is repeated for generating $2^{nR_2}$ sequences $u_2^n(s_2^n,m_2)$, $m_2 \in [1:2^{nR_2}]$. The rates $(R_0,R_1,R_2)$ are chosen so that the ensemble of generated sequences $(u_0^n, u_1^n, u_2^n)$ ``cover'' the set $\Ti_{\epsilon^{\prime}}^{(n)}(U_0,U_1,U_2|s_1^n,s_2^n)$ for all $(s_1^n,s_2^n) \in \Ti_{\epsilon^{\prime}}^{(n)}(S_1,S_2)$. The conditions for ``covering'' are given in the following Lemma.
\begin{lem}
\label{lem}
Define the event
\begin{align*}
\mathcal{A}= & \{(S^n_1,S^n_2,U_0^n(m_0),U_1^n(S^n_1,m_1), \\
             &\qquad \qquad   U_2^n(S^n_2,m_2) \not \in \Ti_{\epsilon^{\prime}}^{(n)}
          \text{ for all } m_0, m_1, m_2\}.
\end{align*}
Then, there exists $\delta(\epsilon^{\prime})$ which tends to zero as $\epsilon^{\prime} \to 0$ such that $\P(\mathcal{A}) \to 0$ as $n \to \infty$
if the following inequalities are satisfied
\begin{align}
R_0 &> I(U_0; S_1,S_2) + \delta(\epsilon^{\prime}), \nonumber\\
R_1 &> I(U_1; S_2|S_1) + \delta(\epsilon^{\prime}), \nonumber\\
R_2 &> I(U_2; S_1|S_2)+ \delta(\epsilon^{\prime}), \nonumber\\
R_0 + R_1 &> I(U_0; S_1,S_2) + I(U_1; U_0,S_2|S_1) + \delta(\epsilon^{\prime}), \nonumber\\
R_0 + R_2 &> I(U_0; S_1,S_2) + I(U_2; U_0,S_1|S_2) + \delta(\epsilon^{\prime}), \nonumber\\
R_1 + R_2 &> I(U_1,S_1; U_2,S_2) - I(S_1; S_2) + \delta(\epsilon^{\prime}),  \nonumber\\
R_0 + R_1 + R_2 &> I(U_0; S_1,S_2) + I(U_1; U_2|U_0,S_1,S_2) \nonumber\\
&\quad + I(U_1; U_0,S_2|S_1) \nonumber\\
&\quad + I(U_2; U_0,S_1|S_2)  + \delta(\epsilon^{\prime}).
\label{eq:cov}
\end{align}
\end{lem}
%\begin{lem}
%\label{lem}
%For each $(s_1^n,s_2^n) \in \Ti_{\epsilon^{\prime}}^{(n)}(S_1,S_2)$, let $u_0^n(m_0)$, $u_1^n(s_1^n,m_1)$ and $u_2^n(s_2^n,m_2)$, $(m_0,m_1,m_2) \in [1:2^{nR_0}]\times [1:2^{nR_1}] \times [1:2^{nR_2}]$, be generated as described above. Then, there exists $\delta(\epsilon)$ which tends to zero as $\epsilon \to 0$ such that
%\begin{align*}
%\P( & S^n_1,S^n_2,U_0^n(m_0),U_1^n(S^n_1,m_1),U_2^n(S^n_2,m_2) \not \in  \Ti_{\epsilon^{\prime}}^{(n)} \text{ for  all }  \nonumber  \\
%&  (m_0,m_1,m_2) \in [1:2^{nR_0}]\times [1:2^{nR_1}] \times [1:2^{nR_2}] \} \to 0
%\end{align*}
%as $n \to \infty$
%if the following inequalities are satisfied
%\begin{equation}
%\label{eq:cov}
%\setlength\arraycolsep{0.2em}
%  \begin{array}{rcl}
%R_0 &>& I(U_0; S_1,S_2) + \delta(\epsilon^{\prime}), \\
%R_1 &>& I(U_1; S_2|S_1) + \delta(\epsilon^{\prime}), \\
%R_2 &>& I(U_2; S_1|S_2)+ \delta(\epsilon^{\prime}), \\
%R_0 + R_1 &>& I(U_0; S_1,S_2) + I(U_1; U_0,S_2|S_1) + \delta(\epsilon^{\prime}), \\
%R_0 + R_2 &>& I(U_0; S_1,S_2) + I(U_2; U_0,S_1|S_2) + \delta(\epsilon^{\prime}), \\
%R_1 + R_2 &>& I(U_1,S_1; U_2,S_2) - I(S_1; S_2) + \delta(\epsilon^{\prime}),  \\
%R_0 + R_1 + R_2 &>& I(U_0; S_1,S_2) + I(U_1; U_2|U_0,S_1,S_2) \\
%& & + I(U_1; U_0,S_2|S_1) + I(U_2; U_0,S_1|S_2)   \\
%& & + \delta(\epsilon^{\prime}).
%\end{array}
%\end{equation}
%\end{lem}
\begin{proof}
See Appendix B.
\end{proof}

\emph{Encoding}:
For each $(s_1^n, s_2^n) \in \aep(S_1,S_2)$, choose a triple $(m_0,m_1,m_2)\in [1:2^{nR_0}]\times [1:2^{nR_1}] \times [1:2^{nR_2}]$ such that $(s_1^n,s_2^n,u_0^n(m_0),u_1^n(s_1^n,m_1),u_2^n(s_2^n,m_2)) \in \Ti_{\epsilon^{\prime}}^{(n)}( S_1,S_2,U_0, U_1, U_2)$. Choose $(m_0,m_1,m_2)=(1,1,1)$ if no such triple can be found. Then, at time $ i\in\{1,\dotsc,n\}$, the encoder transmits $x_i =  x(s_{1i},s_{2i},u_{0i}(m_0),u_{1i}(s_1^n,m_1),u_{2i}(s_2^n,m_2))$. The sequence $x^n$ so generated is the codeword corresponding to the source sequence $(s_1^n, s_2^n)$.

\emph{Decoding}: Let $\epsilon > \epsilon^{\prime}$. Upon observing the sequence $y_1^n$, decoder 1 declares that $s_1^n$ is the transmitted source sequence if this is the unique sequence such that $(s_1^n, u_0^n(m_0), u_1^n(s_1^n,m_1), y_1^n) \in \Ti_{\epsilon}^{(n)}(S_1,U_0,U_1,Y_1)$ for some $m_0 \in [1:2^{nR_0}]$ and $m_1 \in [1:2^{nR_1}]$. Otherwise, declare an error. Similarly, the decoder observing $y_2^n$ declares that $s_2^n$ was sent, if this is the unique sequence such that $(s_2^n, u_0^n(m_0), u_2^n(s_2^n,m_2), y_2^n)\in \Ti_{\epsilon}^{(n)}(S_2,U_0,U_2,Y_2)$ for some $m_0 \in [1:2^{nR_0}]$ and $m_2 \in [1:2^{nR_2}]$.

\emph{Error events}: The detailed analysis of the error probability is
given in Appendix C, but it can be outlined as follows. If the
``covering'' conditions in Lemma \ref{lem} are satisfied, then the
probability of error encoding can be made arbitrarily small by letting
$n \to \infty$. Thus, we focus on the analysis of decoding
errors. Assume that $(m_0,m_1,m_2)=(M_0,M_1,M_2)$ is selected for the
source sequence $(s_1^n, s_2^n)$. The error event for decoder 1 can be
divided into two parts:
\begin{enumerate}
\item There exists a sequence $\hat{s}_1^n \neq S_1^n$ such that $(\hat{s}_1^n, u_0^n(M_0), u_1^n(\hat{s}_1^n,m_1), y_1^n)\in \Ti_{\epsilon}^{(n)}(S_1,U_0,U_1,Y_1) $ for some $m_1 \in [1:2^{nR_1}]$. The probability of this event vanishes as $n \to \infty$ if the following inequality is satisfied
\begin{align}
\label{dec1}
H(S_1) + R_1 &< I(U_1,S_1; U_0,Y_1) - 2\delta(\epsilon).
\end{align}

  \item  There exists a sequence $\hat{s}_1^n \neq S_1^n$ and an $m_0 \neq M_0$ such that $(\hat{s}_1^n, u_0^n(m_0), u_1^n(\hat{s}_1^n,m_1), y_1^n)\in \Ti_{\epsilon}^{(n)}(S_1,U_0,U_1,Y_1) $ for some $m_1 \in [1:2^{nR_1}]$. The probability of this event can be made arbitrarily small as $n \to \infty$ if
\begin{align}
H(S_1) + R_0 + R_1 &< I(U_0,U_1,S_1; Y_1)  \nonumber \\
 & \quad + I(U_0; U_1,S_1) - 3\delta(\epsilon).
\label{dec2}
\end{align}
\end{enumerate}
Similarly, the probability of error for decoder 2 can be made arbitrarily small as $n$ tends to infinity if
\begin{align}
H(S_2) + R_2 &< I(U_2,S_2; U_0,Y_2) - 2\delta(\epsilon), \label{dec3} \\
H(S_2) + R_0 + R_2 &< I(U_0,U_2,S_2; Y_2) \nonumber \\
 & \quad + I(U_0; U_2,S_2) - 3\delta(\epsilon). \label{dec4}
\end{align}

\emph{Fourier--Motzkin elimination}: In summary, by combining
the inequalities in Lemma \ref{lem},~\eqref{dec1},~\eqref{dec2},~\eqref{dec3}, and
\eqref{dec4}, and by letting $\epsilon^{\prime},\epsilon \to 0$, we
conclude that the sources $(S_1,S_2)$ can be reliably transmitted over
the DM-BC $p(y_1,y_2|x)$ if there exists a rate tuple $(R_0,R_1,R_2)$
satisfying
\small
\begin{align}
R_0 &> I(U_0; S_1,S_2) \triangleq v_1,\nonumber\\
R_1 &> I(U_1; S_2|S_1) \triangleq v_2,\nonumber\\
R_2 &> I(U_2; S_1|S_2) \triangleq v_3,\nonumber\\
R_0 + R_1 &> I(U_0; S_1,S_2) + I(U_1; U_0,S_2|S_1) \triangleq v_4,\nonumber\\
R_0 + R_2 &> I(U_0; S_1,S_2) + I(U_2; U_0,S_1|S_2)\triangleq v_5,\nonumber\\
R_1 + R_2 &> I(U_1,S_1; U_2,S_2) - I(S_1; S_2)\triangleq v_6, \nonumber\\
R_0 + R_1 + R_2 &>  I(U_0; S_1,S_2)
 + I(U_1; U_2|U_0,S_1,S_2) \nonumber\\
&\quad  + I(U_1; U_0,S_2|S_1) \nonumber\\
&\quad  + I(U_2; U_0,S_1|S_2)  \triangleq v_7. \nonumber\\
H(S_1) + R_1 &< I(U_1,S_1; U_0,Y_1)\triangleq v_8, \nonumber\\
H(S_1) + R_0 + R_1 &< I(U_0,U_1,S_1; Y_1) + I(U_0; U_1,S_1) \triangleq v_9,\nonumber\\
H(S_2) + R_2 &< I(U_2,S_2; U_0,Y_2) \triangleq v_{10}, \nonumber\\
H(S_2) + R_0 + R_2 &< I(U_0,U_2,S_2; Y_2) + I(U_0; U_2,S_2) \triangleq v_{11}
\label{eq:s1}
\end{align}
\normalsize
for some $p(u_0,u_1,u_2,x|s_1,s_2)$. The final step in our derivation consists of eliminating the auxiliary variables $R_0$, $R_1,$ $R_2$ from the above set of inequalities. This routine task can be performed by application of the standard Fourier--Motzkin elimination algorithm~\cite{HanKobayashi}~\cite{Ziegler}. Unfortunately, brute force application of this algorithm to eliminate $R_0$, $R_1,$ $R_2$ (in this order) results in 28 inequalities involving $H(S_1)$, $H(S_2),$ and $H(S_1, S_2)$, most of which are redundant because implied by other inequalities. Also, to find a minimal set of non-redundant inequalities, one has to verify linear information inequalities involving several random variables. Therefore, although conceptually simple, application of the Fourier--Motzkin elimination algorithm can be tedious.

We would like to illustrate a technique that allow us to get around these difficulties and to efficiently perform the elimination using a computer program.  The key idea is to treat the right hand sides of \eqref{eq:s1} as auxiliary variables. Denote them by $v_1,\dotsc,v_{11}$. Next, observe that
\begin{equation}
\label{eq:s3}
\setlength\arraycolsep{0.2em}
  \begin{array}{rcl}
v_1 &\ge & 0, \\
v_2 &\ge & 0, \\
v_3 &\ge & 0, \\
v_1 + v_2 &\le & v_4, \\
v_1 + v_3 &\le & v_5, \\
v_2 + v_3 &\le & v_6, \\
v_4 + v_6 &\le & v_2 + v_7,\\
v_4 + v_5 &\le & v_1 + v_7, \\
v_5 + v_6 &\le & v_3 + v_7,  \\
v_8 & \le & v_9, \\
v_{10} & \le & v_{11}.
  \end{array}
\end{equation}
The above information inequalities can be checked directly using the chain rule and the nonnegativity of mutual information or, alternatively, they can be verified by the software ITIP~\cite{ITIP}.

Combining \eqref{eq:s1} and \eqref{eq:s3} we obtain a system of linear inequalities involving the variables $(H(S_1),H(S_2),R_0,R_1,R_2,v_1,v_2,\dotsc,v_{11})$. Next, we eliminate the auxiliary variables $R_0$, $R_1$ and $R_2$ from this system of linear inequalities. Since \eqref{eq:s1} and \eqref{eq:s3} have constant coefficients, the  Fourier--Motzkin elimination can now be performed by a computer program, e.g., by the software PORTA~\cite{PORTA}. The algorithm results in the following inequalities involving $H(S_1)$ and $H(S_2)$:
\begin{align}
H(S_1) &< v_9 -v_4, \nonumber\\
H(S_1) &< v_8 -v_2, \nonumber\\
H(S_2) &< v_{11} - v_5, \nonumber\\
H(S_2) &< v_{10} - v_3, \nonumber\\
H(S_1)+H(S_2) &< v_9+v_{10}-v_7, \nonumber\\
H(S_1)+H(S_2) &< v_8+v_{11}-v_7, \nonumber\\
H(S_1)+H(S_2) &< v_8+v_{10}-v_6, \nonumber\\
H(S_1)+H(S_2) &< v_9+v_{11}-v_7-v_1.
\label{eq:fm}
\end{align}
To complete the proof, we argue that three of the above inequalities can be discarded because inactive. We proceed as follows. Define three new auxiliary random variables $\tilde{U}_i = (U_i, W)$, $i \in \{0,1,2\}$, where $W$ is chosen independent of everything else, and replace $U_i$ with $\tilde{U}_i$ in \eqref{eq:s1}. Substituting the new values of $v_1,\dotsc,v_{11}$, and performing some manipulations, we can rewrite \eqref{eq:fm} as follows
\begin{align}
H(S_1)    & <  I(U_0, U_1, S_1; Y_1) - I(U_0, U_1; S_2|S_1),                         \nonumber  \\
H(S_1)     & <  I(U_1,S_1; U_0,Y_1) \nonumber  \\
& \quad - I(U_1; S_2|S_1) + H(W),            \label{eq:el1}  \\
H(S_2)    & <   I(U_0, U_2, S_2; Y_2) - I(U_0, U_2; S_1|S_2), \nonumber  \\
H(S_2)    & <  I(U_2,S_2; U_0,Y_2) \nonumber  \\
&  \quad  - I(U_2; S_1|S_2) + H(W), \label{eq:el2}  \\
H(S_1,S_2)  &< I(U_0, U_1, S_1; Y_1) + I(U_2, S_2; Y_2|U_0), \nonumber \\
&  \quad               - I(U_1, S_1; U_2, S_2|U_0), \nonumber \\
H(S_1,S_2)  &< I(U_1, S_1;Y_1|U_0) + I(U_0, U_2, S_2; Y_2) \nonumber \\
&  \quad               - I(U_1, S_1; U_2, S_2|U_0), \nonumber  \\
H(S_1,S_2)& <  I(U_2,S_2; U_0,Y_2) + I(U_1,S_1; U_0,Y_1) \nonumber \\
&  \quad   - I(U_1,S_1; U_2,S_2) + H(W), \label{eq:el3} \\
H(S_1,S_2)  & <  I(U_0, U_1, S_1; Y_1) + I(U_0, U_2, S_2; Y_2) \nonumber \\
&  \quad      - I(U_1, S_1; U_2, S_2|U_0) - I(S_1, S_2; U_0). \nonumber
\end{align}
By letting $H(W) \to \infty$, we can make \eqref{eq:el1},
\eqref{eq:el2} and \eqref{eq:el3} trivial, and thus inactive, while
leaving the remaining inequalities unchanged. In summary, the five
inequalities remaining from the Fourier--Motzkin elimination are those
appearing in the statement of Theorem \ref{thm1}. This proves the
desired sufficiency.

\section{Discussion}

The best known inner bound for the problem of broadcasting correlated
sources over a general DM-BC is due to Han and Costa with the
associated coding scheme that cleverly combines a m\'{e}lange of
coding techniques, including joint typicality encoding and decoding,
random hashing, superposition coding, and the use of the common part
between two random variables. To investigate which techniques are
crucial, this paper presents a ``minimalistic'' coding scheme in which
we remove from the coding scheme of Han and Costa all unnecessary
components without affecting the overall performance. Our proposed
coding scheme does not require random hashing and superposition
coding, and it does not involve the common part of two random
variables. This highlights how source encoding can be performed by
simply jointly ``covering'' the set of typical sources using auxiliary
correlated random variables. An interesting implication of our result,
albeit no more than theoretically amusing, is that the capacity of the
degraded DM-BC can be achieved without employing Cover's superposition
coding~\cite{Cover_sc}.

We would like to conclude the paper by showing that the sufficient
condition stated in Theorem~\ref{thm1} can also be proved by means of
a more elaborate coding scheme which uses superposition coding in
addition to joint typicality encoding and decoding. This alternative
scheme differs from the one described in Section~IV in the way the
random codebook is generated, as the auxiliary random variable $U_0$
serves as a cloud center for generating the auxiliary random variables
$U_1$ and $U_2$. Appendix D presents a description of the code
construction and a sketch of the analysis of the associated
probability of error.

\appendices

\section{Proof of Proposition \ref{prop}}
\label{sec:appendix1}

Suppose that the source  $(S_1,S_2)$ satisfies the sufficient
condition stated in Theorem \ref{thm1}, so that the inequalities
 in Theorem \ref{thm1} are satisfied for
some triple $(U_0,U_1,U_2)$. The goal is to show that $(S_1,S_2)$ also
satisfy the inequalities in Theorem \ref{thmHC} for the same triple of
auxiliary random variables. Expanding the left hand side of
\eqref{KM1}, we obtain that
%%{\small
\begin{align*}
& H(S_1, S_2) \nonumber \\
&\; < I(U_0, U_1, S_1; Y_1) + I(U_2, S_2; Y_2|U_0) \nonumber \\
&\quad  - I(U_1, S_1; U_2, S_2|U_0), \\
&\; = H(Y_1) - H(Y_1|U_0,U_1,S_1) -  H(S_2,U_2|Y_2,U_0) \nonumber \\
&\; \quad  + H(S_2,U_2|U_0,U_1,S_1) \nonumber \\
&\;\leq H(Y_1) - H(Y_1|K,U_0,U_1,S_1) -  H(S_2,U_2|Y_2,K,U_0) \nonumber \\
&\; \quad  + H(S_2,U_2|K,U_0,U_1,S_1), \\
&\; = I(K,U_0, U_1, S_1; Y_1) + I(U_2, S_2; Y_2|K,U_0) \nonumber \\
&\;  \quad               - I(U_1, S_1; U_2, S_2|K,U_0),
\end{align*}
%%}
%%{\normalsize
where the second inequality follows from the fact that conditioning reduces the entropy and that $K=f(S_1)=g(S_2)$ is a deterministic function of the sources $(S_1,S_2)$. Thus, $(S_1,S_2)$ satisfy \eqref{HC1}. Proceeding in a similar way, it is immediate to show that $(S_1,S_2)$ satisfy \eqref{HC2}. Finally, we expand the left hand side of \eqref{KM3} and obtain that
%%}
%%{\small
\begin{align*}
& H(S_1, S_2) \\
& < I(U_0, U_1, S_1; Y_1) + I(U_0, U_2, S_2; Y_2)  \nonumber \\
& \quad  - I(U_1, S_1; U_2, S_2|U_0) - I(S_1, S_2; U_0). \nonumber \\
& = H(Y_1) - H(Y_1|U_0,U_1,S_1) -  H(Y_2) - H(Y_2|U_0,U_2,S_2) \nonumber \\
& \quad   - H(U_1|S_1,U_0) + H(U_1,S_1|U_0,U_2,S_2) - H(S_1) \nonumber \\
& \leq H(Y_1) - H(Y_1|K,U_0,U_1,S_1) -  H(Y_2) \nonumber \\
&\quad - H(Y_2|K,U_0,U_2,S_2) - H(U_1|S_1,K,U_0) \nonumber \\
&\quad + H(U_1,S_1|K,U_0,U_2,S_2) - H(S_1) \\
& = I(K,U_0, U_1, S_1; Y_1) + I(K,U_0, U_2, S_2; Y_2)   \nonumber \\
&  \quad   - I(U_1, S_1; U_2, S_2|K,U_0) - I(S_1, S_2; K,U_0), \nonumber
\end{align*}
%%}
so the sources $(S_1,S_2)$ satisfy \eqref{HC3}. Thus, we conclude that
$(S_1,S_2)$ also satisfy the sufficient condition stated in
Theorem \ref{thmHC}.

\section{Proof of Lemma \ref{lem}}
\label{sec:appendix2}
The technique used in the proof is similar to the one in~\cite{ElGamal}. For each $(s_1^n,s_2^n)\in \Ti_{\epsilon^{\prime}}^{(n)}(S_1,S_2)$, let
\begin{align*}
&\mathcal{A}(s_1^n,s_2^n) \\
& \triangleq \{ (m_0,m_1,m_2) \in [1:2^{nR_0}]\times [1:2^{nR_1}] \times [1:2^{nR_2}]: \\
& \quad  (s^n_1,s^n_2,U_0^n(m_0),U_1^n(s^n_1,m_1),U_2^n(s^n_2,m_2)) \in \Ti_{\epsilon^{\prime}}^{(n)} \}.
\end{align*}
Then,
\begin{align*}
\P(\mathcal{A}) & \le \P((s_1^n,s_2^n) \not \in \Ti_{\epsilon^{\prime}}^{(n)}) \\
 & \quad +\sum_{(s_1^n,s_2^n)\in \Ti_{\epsilon^{\prime}}^{(n)}} p(s_1^n,s_2^n) \P( |\mathcal{A}(s_1^n,s_2^n)|=0).
\end{align*}
By the law of large number the first term in the sum tends to 0 as $n \to \infty$. To bound the second term, observe that
\begin{align}
\P(|\mathcal{A}(s_1^n,s_2^n)|=0) \le  \frac{\text{Var}( |\mathcal{A}(s_1^n,s_2^n)|  )}{(\mathsf{E}[|\mathcal{A}(s_1^n,s_2^n)|]) ^2}.
\label{eq:ineq}
\end{align}
%\begin{align}
%&\P(|\mathcal{A}(s_1^n,s_2^n)|=0) \\
%& \le  \P(|\mathcal{A}(s_1^n,s_2^n)|=0) + \P(|\mathcal{A}| \ge 2 \mathbb{E}|\mathcal{A}(s_1^n,s_2^n)| )\P(|\mathcal{A}(s_1^n,s_2^n)|>0) \nonumber \\
%& =  \P( (|\mathcal{A}(s_1^n,s_2^n)| - \mathbb{E}|\mathcal{A}(s_1^n,s_2^n)| )^2  \ge (\mathbb{E}|\mathcal{A}(s_1^n,s_2^n)| )^2) \nonumber \\
%& \le  \frac{\text{Var}( |\mathcal{A}(s_1^n,s_2^n)|  )}{(\mathbb{E}|\mathcal{A}(s_1^n,s_2^n)| )^2}.
%\label{eq:ineq}
%\end{align}
Using indicator random variables, we express $|\mathcal{A}(s_1^n,s_2^n)|$ as
\begin{align*}
|\mathcal{A}(s_1^n,s_2^n)| = \sum_{m_0,m_1,m_2} \mathbbm{1}_{\mathcal{A}_{m_0,m_1,m_2}}
\end{align*}
where
\begin{align*}
\mathcal{A}_{m_0,m_1,m_2} & \triangleq \{ (s^n_1,s^n_2,U_0^n(m_0),U_1^n(s^n_1,m_1),U_2^n(s^n_2,m_2)) \\
& \quad \quad \in \Ti_{\epsilon^{\prime}}^{(n)}\}.
\end{align*}
Then, taking expectations and using the fact that codewords are independently generated,
\begin{align}
\label{eq:mean}
\mathsf{E} [|\mathcal{A}(s_1^n,s_2^n)|] & = \sum_{m_0,m_1,m_2} \P(\mathcal{A}_{m_0,m_1,m_2})  \nonumber \\
& = 2^{n(R_0+R_1+R_2)} \P(\mathcal{A}_{1,1,1})
\end{align}
Similarly, we have
\small
\begin{align*}
& \mathsf{E} [|\mathcal{A}(s_1^n,s_2^n)|^2] \\
& = \sum_{m_0,m_1,m_2} \P(\mathcal{A}_{m_0,m_1,m_2}) \\
& \quad + \sum_{m_0,m_1,m_2} \sum_{m_0^\prime \not = m_0} \P(\mathcal{A}_{m_0,m_1,m_2},\mathcal{A}_{m_0^\prime,m_1,m_2}) \\
& \quad  + \sum_{m_0,m_1,m_2} \sum_{m_1^\prime \not = m_1} \P(\mathcal{A}_{m_0,m_1,m_2},\mathcal{A}_{m_0,m_1^\prime,m_2}) \\
& \quad  + \sum_{m_0,m_1,m_2} \sum_{m_2^\prime \not = m_2} \P(\mathcal{A}_{m_0,m_1,m_2},\mathcal{A}_{m_0,m_1,m_2^\prime}) \\
& \quad  + \sum_{m_0,m_1,m_2} \sum_{ \substack{m_0^\prime \not = m_0 \\ m_1^\prime \not = m_1} } \P(\mathcal{A}_{m_0,m_1,m_2},\mathcal{A}_{m_0^\prime,m_1^\prime,m_2}) \\
& \quad  + \sum_{m_0,m_1,m_2} \sum_{ \substack{m_0^\prime \not = m_0 \\ m_2^\prime \not = m_2} } \P(\mathcal{A}_{m_0,m_1,m_2},\mathcal{A}_{m_0^\prime,m_1,m_2^\prime}) \\
& \quad  + \sum_{m_0,m_1,m_2} \sum_{ \substack{m_1^\prime \not = m_1 \\ m_2^\prime \not = m_2} } \P(\mathcal{A}_{m_0,m_1,m_2},\mathcal{A}_{m_0,m_1^\prime,m_2^\prime}) \\
& \quad  + \sum_{m_0,m_1,m_2} \sum_{ \substack{m_0^\prime \not = m_0 \\ m_1^\prime \not = m_1 \\ m_2^\prime \not = m_2 }} \P(\mathcal{A}_{m_0,m_1,m_2},\mathcal{A}_{m_0^\prime,m_1^\prime,m_2^\prime}).
\end{align*}
\normalsize
Since the codewords are independently generated,  we can re-write the above equality as follows
\begin{align*}
\mathsf{E} |\mathcal{A}|^2 & \le 2^{n(R_0+R_1+R_2)} \P(\mathcal{A}_{1,1,1}) \\
& \quad + 2^{n(2R_0+R_1+R_2)} \P(\mathcal{A}_{1,1,1},\mathcal{A}_{2,1,1}) \\
& \quad  + 2^{n(R_0+2R_1+R_2)} \P(\mathcal{A}_{1,1,1},\mathcal{A}_{1,2,1}) \\
& \quad  + 2^{n(R_0+R_1+2R_2)} \P(\mathcal{A}_{1,1,1},\mathcal{A}_{1,1,2}) \\
& \quad  + 2^{n(2R_0+2R_1+R_2)} \P(\mathcal{A}_{1,1,1},\mathcal{A}_{2,2,1}) \\
& \quad  + 2^{n(2R_0+R_1+2R_2)} \P(\mathcal{A}_{1,1,1},\mathcal{A}_{2,1,2}) \\
& \quad  + 2^{n(R_0+2R_1+2R_2)} \P(\mathcal{A}_{1,1,1},\mathcal{A}_{1,2,2}) \\
& \quad  + 2^{n(2R_0+2R_1+2R_2)} \P(\mathcal{A}_{1,1,1},\mathcal{A}_{2,2,2}).
\end{align*}
It is easily seen that $\P(\mathcal{A}_{1,1,1},\mathcal{A}_{2,2,2})=\P(\mathcal{A}_{1,1,1})^2$. It follows that
\begin{align}
\text{Var}(|\mathcal{A}|) & \le 2^{n(R_0+R_1+R_2)} \P(\mathcal{A}_{1,1,1}) \nonumber \\
& \quad  + 2^{n(2R_0+R_1+R_2)} \P(\mathcal{A}_{1,1,1},\mathcal{A}_{2,1,1}) \nonumber  \\
& \quad  + 2^{n(R_0+2R_1+R_2)} \P(\mathcal{A}_{1,1,1},\mathcal{A}_{1,2,1}) \nonumber  \\
& \quad  + 2^{n(R_0+R_1+2R_2)} \P(\mathcal{A}_{1,1,1},\mathcal{A}_{1,1,2}) \nonumber  \\
& \quad  + 2^{n(2R_0+2R_1+R_2)} \P(\mathcal{A}_{1,1,1},\mathcal{A}_{2,2,1}) \nonumber  \\
& \quad  + 2^{n(2R_0+R_1+2R_2)} \P(\mathcal{A}_{1,1,1},\mathcal{A}_{2,1,2}) \nonumber  \\
& \quad  + 2^{n(R_0+2R_1+2R_2)} \P(\mathcal{A}_{1,1,1},\mathcal{A}_{1,2,2}).
\label{eq:var}
\end{align}
By the joint typicality lemma~\cite{ElGamalKim}, we have
\begin{align*}
\P(\mathcal{A}_{2,1,1},\mathcal{A}_{1,1,1}) & \le 2^{-n [I(U_0;S_1,S_2,U_1,U_2) +I(U_1;U_0,S_2|S_1) ]} \nonumber \\
& \quad \cdot 2^{-n[ I(U_0;S_1,S_2)+ I(U_2;U_0,U_1,S_1|S_2)] - \delta(\epsilon^{\prime}) ] },  \\
%%%
\P(\mathcal{A}_{1,2,1},\mathcal{A}_{1,1,1}) & \le 2^{-n [I(U_1;S_2,U_0,U_2|S_1) +I(U_1;U_0,S_2|S_1) ]} \nonumber \\
& \quad \cdot 2^{-n[ I(U_0;S_1,S_2)+ I(U_2;U_0,U_1,S_1|S_2)] - \delta(\epsilon^{\prime}) ] },  \\
%%%
\P(\mathcal{A}_{1,1,2},\mathcal{A}_{1,1,1}) & \le 2^{-n [I(U_2;S_1,U_0,U_1|S_2) +I(U_1;U_0,S_2|S_1) ]} \nonumber \\
& \quad \cdot 2^{-n[ I(U_0;S_1,S_2)+ I(U_2;U_0,U_1,S_1|S_2)] - \delta(\epsilon^{\prime}) ] },  \\
%%%
\P(\mathcal{A}_{2,2,1},\mathcal{A}_{1,1,1}) &\le 2^{-n [I(U_0;S_1,S_2,U_1,U_2) + I(U_1;S_2,U_2|S_1)] } \nonumber \\
& \quad \cdot 2^{-n[ I(U_1;U_0,S_2|S_1)  + I(U_0;S_1,S_2) ] }  \\
& \quad \cdot 2^{-n[ I(U_2;U_0,U_1,S_1|S_2)] - \delta(\epsilon^{\prime}) ]},  \\
%%%
\P(\mathcal{A}_{2,1,2},\mathcal{A}_{1,1,1}) &\le 2^{-n [I(U_0;S_1,S_2,U_1,U_2) + I(U_2;S_1,U_1|S_2) ]} \nonumber \\
& \quad \cdot 2^{-n[ I(U_1;U_0,S_2|S_1)  + I(U_0;S_1,S_2) ] }  \\
& \quad \cdot 2^{-n[ I(U_2;U_0,U_1,S_1|S_2)] - \delta(\epsilon^{\prime}) ]},  \\
%%%
\P(\mathcal{A}_{1,2,2},\mathcal{A}_{1,1,1}) &\le 2^{-n [I(U_1;U_0,S_2|S_1) + I(U_2;U_0,U_1,S_1|S_2) ]} \nonumber \\
& \quad \cdot 2^{-n[ I(U_1;U_0,S_2|S_1)  + I(U_0;S_1,S_2) ] }  \\
& \quad \cdot 2^{-n[ I(U_2;U_0,U_1,S_1|S_2)] - \delta(\epsilon^{\prime}) ]},\\
%%\end{align*}
%%and
%%\begin{align*}
\P(\mathcal{A}_{1,1,1}) & \ge 2^{-n [I(U_0;S_1,S_2) +I(U_1;U_0,S_2|S_1) } \nonumber \\
& \quad \cdot 2^{-n[ + I(U_2;U_0,U_1,S_1|S_2)] + \delta(\epsilon^{\prime}) ] }.
\end{align*}
%\begin{align*}
%\P(\mathcal{A}_{2,1,1}|\mathcal{A}_{1,1,1}) & \le 2^{n [H(U_0|S_1,S_2,U_1,U_2) - H(U_0) -\delta(\epsilon^{\prime}) ] } \\
%\P(\mathcal{A}_{1,2,1}|\mathcal{A}_{1,1,1}) & \le 2^{n[ H(U_1|S_1,S_2,U_0,U_2) -H(U_1|S_1) -\delta(\epsilon^{\prime}) ]} \\
%\P(\mathcal{A}_{1,1,2}|\mathcal{A}_{1,1,1}) & \le 2^{n[ H(U_2|S_1,S_2,U_0,U_1) -H(U_2|S_2) -\delta(\epsilon^{\prime}) ]} \\
%\P(\mathcal{A}_{2,2,1}|\mathcal{A}_{1,1,1}) &\le 2^{n[ H(U_0,U_1|S_1,S_2,U_2) -H(U_0) ]} \nonumber \\
%& \quad \cdot 2^{n[ - H(U_1|S_1) -\delta(\epsilon^{\prime}) ]}  \\
%\P(\mathcal{A}_{2,1,2}|\mathcal{A}_{1,1,1}) &\le 2^{n[ H(U_0,U_2|S_1,S_2,U_1) -H(U_0) ]} \nonumber \\
%& \quad \cdot 2^{n[ - H(U_2|S_2) -\delta(\epsilon^{\prime}) ]}  \\
%\P(\mathcal{A}_{1,2,2}|\mathcal{A}_{1,1,1}) &\le 2^{n[ H(U_1,U_2|S_1,S_2,U_0) -H(U_1|S_1) ]} \nonumber \\
%& \quad \cdot 2^{n[ - H(U_2|S_2) -\delta(\epsilon^{\prime}) ]}
%\end{align*}
%and
%\begin{align*}
%\P(\mathcal{A}_{1,1,1}) & \ge 2^{n [H(U_0,U_1,U_2|S_1,S_2) - H(U_0) -H(U_1|S_1) ]} \nonumber \\
%& \quad \cdot 2^{n[ - H(U_2|S_2) +\delta(\epsilon^{\prime}) ] }.
%\end{align*}
Substituting the above inequalities into \eqref{eq:mean} and \eqref{eq:var}, and making use of \eqref{eq:ineq}, we obtain that $\P(|\mathcal{A}(s_1^n,s_2^n)|=0) \rightarrow 0$ as $n \rightarrow \infty$ if  conditions \eqref{eq:cov} are simultaneously satisfied. This completes the proof of the lemma.

\section{Analysis of the Probability of Error}
\label{sec:appendix3}

Assume that $(m_0,m_1,m_2)=(M_0,M_1,M_2)$ for the transmitted source sequences  $(s_1^n, s_2^n)$.  To study the error probability for decoder 1, define the events
\begin{align}
\mathcal{E}_1 & \triangleq \{ (S_1^n,U_0^n(M_0),X_1^n(S_1^n,M_0),Y_1^n)  \not \in   \Ti_\epsilon^{(n)} \}, \nonumber \\
\mathcal{E}_2 & \triangleq \{ (\tilde{s}_1^n,U_0^n(M_0),U_1^n(\tilde{s}_1^n,m_1),Y_1^n) \in   \Ti_\epsilon^{(n)}  \nonumber \\
 & \quad \quad  \text{ for some } \tilde{s}_1^n \not = S_1^n \text{ and } m_1\in[1:2^{nR_1}] \}, \nonumber \\
\mathcal{E}_3 & \triangleq \{ (\tilde{s}_1^n,U_0^n(m_0),U_1^n(\tilde{s}_1^n,m_1),Y_1^n) \in   \Ti_\epsilon^{(n)}  \nonumber \\
 & \quad \quad  \text{ for some } \tilde{s}_1^n \not = S_1^n, m_0 \not = M_0, \text{ and } m_1\in[1:2^{nR_1}] \}. \nonumber
\end{align}
The average probability of decoding error at decoder 1 is bounded by
\begin{equation*}
\P(\hat{S}_1^n \not = S_1^n) \le \P(\mathcal{E}_1) + \P(\mathcal{E}_2) + \P(\mathcal{E}_3).
\end{equation*}
By the law of large numbers, $\P(\mathcal{E}_1)\rightarrow 0$ as $n \rightarrow \infty$. Next, consider the second term. By the union bound, we have
\begin{align*}
& \P(\mathcal{E}_2) \\
& = \sum_{s_1^n} p(s_1^n) P\{ (\tilde{s}_1^n,U_0^n(M_0),U_1^n(\tilde{s}_1^n,m_1),Y_1^n) \in   \Ti_\epsilon^{(n)} \\
 & \quad\quad
   \text{ for some } \tilde{s}_1^n \not = S_1^n, \text{ and } m_1 \in[1:2^{nR_1}] \mid S_1^n = s^n_1 \} \\
& \le \sum_{s_1^n} p(s_1^n) \sum_{\tilde{s}_1^n \not = s_1^n} 
\sum_{m_1=1}^{2^{nR_1}} P\{ (\tilde{s}_1^n,U_0^n(M_0),U_1^n(\tilde{s}_1^n,m_1),\nonumber \\
 & \quad\quad\quad\quad\quad\quad\quad\quad\quad\quad
\quad \quad  Y_1^n) \in   \Ti_\epsilon^{(n)} \mid S_1^n = s^n_1 \}.
\end{align*}
Conditioned on $S_1^n = s^n_1$, for all $\tilde{s}_1^n \not = s_1^n$ and for all $m_1 \in[1:2^{nR_1}]$, we have that
%\begin{align*}
$(U_0^n(M_0),U_1^n(\tilde{s}_1^n,m_1)$ $,Y_1^n)  \sim p(u_1^n( \tilde{s}_1^n,m_1) |\tilde{s}_1^n)$ $p( u_0^n(M_0), y_1^n|s^n_1)$.
%\end{align*}
Thus,
\begin{align*}
& \P(\mathcal{E}_2) \\
& \le \sum_{s_1^n} p(s_1^n)
\sum_{\tilde{s}_1^n \not = s_1^n, (\tilde{s}_1^n,u_0^n,u_1^n,y_1^n ) \in \Ti_\epsilon^{(n)}  } \sum_{m_1=1}^{2^{nR_1}} p( u_0^n(M_0), y_1^n|s^n_1) \nonumber \\
 & \quad \cdot  p(u_1^n( \tilde{s}_1^n,m_1) |\tilde{s}_1^n) \\
& = \sum_{ (\tilde{s}_1^n,u_0^n,u_1^n,y_1^n ) \in \Ti_\epsilon^{(n)} }  \sum_{s_1^n \not = \tilde{s}_1^n} \sum_{m_1=1}^{2^{nR_1}} p( u_0^n(M_0), y_1^n|s^n_1)  \nonumber \\
 & \quad \cdot p(u_1^n( \tilde{s}_1^n,m_1) |\tilde{s}_1^n) p(s_1^n) \\
& \le \sum_{ (\tilde{s}_1^n,u_0^n,u_1^n,y_1^n ) \in \Ti_\epsilon^{(n)} }  \sum_{s_1^n} \sum_{m_1=1}^{2^{nR_1}} p( u_0^n(M_0), y_1^n|s^n_1) \nonumber \\
 & \quad \cdot p(u_1^n( \tilde{s}_1^n,m_1) |\tilde{s}_1^n) p(s_1^n) \\
& = \sum_{ (\tilde{s}_1^n,u_0^n,u_1^n,y_1^n ) \in \Ti_\epsilon^{(n)} }  \sum_{m_1=1}^{2^{nR_1}} p( u_0^n(M_0), y_1^n) p(u_1^n( \tilde{s}_1^n,m_1) |\tilde{s}_1^n) \\
& \le \sum_{ (\tilde{s}_1^n,u_0^n,u_1^n,y_1^n ) \in \Ti_\epsilon^{(n)} } 2^{n R_1} 2^{-n(H(U_0,Y_1) + H(U_1|S_1) -2\delta(\epsilon) )} \\
& \le 2^{-n(H(S_1,U_0,U_1,Y_1)+\delta(\epsilon) ) + n R_1
-n(H(U_0,Y_1) + H(U_1|S_1) -2\delta(\epsilon) )}.
\end{align*}
Collecting the entropy terms at the exponent, we have
\begin{align*}
& H(S_1,U_0,U_1,Y_1) - H(U_0,Y_1) - H(U_1|S_1) \\
& = H(S_1) +  H(U_1|S_1) + H(U_0,Y_1|U_1,S_1) \\
& \quad - H(U_0,Y_1) - H(U_1|S_1) \\
& = H(S_1) - I(U_1,S_1; U_0,Y_1).
\end{align*}
Thus $\P(\mathcal{E}_2) \rightarrow 0$ as $n \rightarrow \infty$ if
\begin{align}
\label{eq:proof1}
H(S_1) + R_1 &< I(U_1,S_1; U_0,Y_1) -2\delta(\epsilon).
\end{align}
Finally, consider the third term. By the union bound, we have
\begin{align*}
& \P(\mathcal{E}_3) \\
& = \sum_{s_1^n} p(s_1^n) P\{ (\tilde{s}_1^n,U_0^n(m_0),U_1^n(\tilde{s}_1^n,m_1),Y_1^n) \in   \Ti_\epsilon^{(n)} \\
 & \quad \quad \text{ for some } \tilde{s}_1^n \not = S_1^n, m_0 \not = M_0, \text{ and } m_1 | S_1^n = s^n_1 \} \\
& \le \sum_{s_1^n} p(s_1^n) \sum_{\tilde{s}_1^n \not = s_1^n} \sum_{m_0=1}^{2^{nR_0}} \sum_{m_1=1}^{2^{nR_1}} P\{ (\tilde{s}_1^n,U_0^n(m_0),U_1^n(\tilde{s}_1^n,m_1),\nonumber \\
 & \quad \quad  Y_1^n) \in   \Ti_\epsilon^{(n)} | S_1^n = s^n_1 \}.
\end{align*}
Conditioned on $S_1^n = s^n_1$, for all $\tilde{s}_1^n \not = s_1^n$, $m_0 \not = M_0$, and for all $m_1$, we have that
%\begin{align*}
$(U_0^n(m_0),U_1^n(\tilde{s}_1^n,m_1),Y_1^n)  \sim p(u_0^n(m_0)) p(u_1^n( \tilde{s}_1^n,m_1) |\tilde{s}_1^n)p( y_1^n|s^n_1)$.
%\end{align*}
Thus,
\begin{align*}
& \P(\mathcal{E}_3) \\
& \le \sum_{s_1^n} p(s_1^n)
\sum_{\tilde{s}_1^n \not = s_1^n, (\tilde{s}_1^n,u_0^n,u_1^n,y_1^n ) \in \Ti_\epsilon^{(n)}  } \sum_{m_0=1}^{2^{nR_0}} \sum_{m_1=1}^{2^{nR_1}} p(u_0^n(m_0)) \nonumber \\
 & \quad \cdot p(u_1^n( \tilde{s}_1^n,m_1) |\tilde{s}_1^n)p( y_1^n|s^n_1) \\
& = \sum_{ (\tilde{s}_1^n,u_0^n,u_1^n,y_1^n ) \in \Ti_\epsilon^{(n)} }  \sum_{s_1^n \not = \tilde{s}_1^n} \sum_{m_0=1}^{2^{nR_0}} \sum_{m_1=1}^{2^{nR_1}} p(u_0^n(m_0)) \nonumber \\
 & \quad \cdot  p(u_1^n( \tilde{s}_1^n,m_1) |\tilde{s}_1^n)p( y_1^n|s^n_1) p(s_1^n) \\
& \le \sum_{ (\tilde{s}_1^n,u_0^n,u_1^n,y_1^n ) \in \Ti_\epsilon^{(n)} }  \sum_{s_1^n} \sum_{m_0=1}^{2^{nR_0}} \sum_{m_1=1}^{2^{nR_1}} p(u_0^n(m_0)) \nonumber \\
 & \quad \cdot p(u_1^n( \tilde{s}_1^n,m_1) |\tilde{s}_1^n)p( y_1^n|s^n_1) p(s_1^n) \\
& = \sum_{ (\tilde{s}_1^n,u_0^n,u_1^n,y_1^n ) \in \Ti_\epsilon^{(n)} }  \sum_{m_0=1}^{2^{nR_0}} \sum_{m_1=1}^{2^{nR_1}} p(u_0^n(m_0)) \nonumber \\
 & \quad \cdot  p(u_1^n( \tilde{s}_1^n,m_1) |\tilde{s}_1^n)p( y_1^n) \\
& \le \sum_{ (\tilde{s}_1^n,u_0^n,u_1^n,y_1^n )\in \Ti_\epsilon^{(n)} } 2^{n[R_0+R_1 -(H(U_0) + H(U_1|S_1) + H(Y_1) -3\delta(\epsilon) )]} \\
& \le 2^{-n(H(S_1,U_0,U_1,Y_1)+\delta(\epsilon) )} 2^{n(R_0+R_1)} \nonumber \\
 & \quad \cdot 2^{-n(H(U_0) + H(U_1|S_1) + H(Y_1) -3\delta(\epsilon) )}.
\end{align*}
Collecting the entropy terms at the exponent, we have
\begin{align*}
& H(S_1,U_0,U_1,Y_1) - H(U_0) - H(U_1|S_1) - H(Y_1)  \\
& = H(S_1) +  H(U_1|S_1) + H(U_0|U_1,S_1) + H(Y_1| S_1,U_0,U_1) \\
& \quad - H(U_0) - H(U_1|S_1) - H(Y_1) \\
& = H(S_1) - I(U_0,U_1,S_1; Y_1) - I(U_0; U_1,S_1).
\end{align*}
Thus $\P(\mathcal{E}_3) \rightarrow 0$ as $n \rightarrow \infty$ if
\begin{align}
\label{eq:proof2}
H(S_1) + R_0 + R_1 &< I(U_0,U_1,S_1; Y_1) + \nonumber \\
 & \qquad + I(U_0; U_1,S_1) -3\delta(\epsilon).
\end{align}
In summary, the probability of error for decoder 1 can be made arbitrarily small by letting $n \to \infty$ if \eqref{eq:proof1} and \eqref{eq:proof2} hold. Finally, the study of the error probability for decoder 2 follows from a similar argument.

\section{An Alternative Coding Scheme Using Superposition Coding}
\label{sec:appendix4}

We describe here the construction of coding scheme based on
superposition coding which yields the same sufficient condition as the
one stated in Theorem \ref{thm1}.

\emph{Random codebook generation}: Let $\epsilon^{\prime}>0$. Fix a joint distribution $p(u_0,u_1,u_2|s_1,s_2)$ and, without loss of generality, let $p(x|u_0,u_1,u_2,s_1,s_2)$ be a chosen deterministic function $x(s_1,s_2,u_0,u_1,u_2)$. Compute $p(u_0)$, $p(u_1|u_0,s_1)$ and $p(u_2|u_0,s_2)$ for the given source distribution $p(s_1,s_2)$.  Randomly and independently generate $2^{nR_0}$ sequences $u_0^n(m_0)$, $m_0 \in [1:2^{nR_0}]$, each according to $\prod_{i=1}^n p_{U_0}(u_{0i})$. For each source sequence $s_1^n$ and $u_0^n(m_0)$, randomly and independently generate $2^{nR_1}$ sequences $u_1^n(s_1^n,m_0,m_1)$, $m_1 \in [1:2^{nR_1}]$, each according to
$\prod_{i=1}^n p_{U_1|S_1,U_0}(u_{1,i}|s_{1i},u_{0,i}(m_0))$. Similarly, for each source sequence $s_2^n$ and $u_0^n(m_0)$, randomly and independently generate $2^{nR_2}$ sequences
$u_2^n(s_2^n,m_0,m_2)$, $m_2 \in [1:2^{nR_2}]$, each according to
$\prod_{i=1}^n p_{U_2|S_2,U_0}(u_{2,i}|s_{2,i},u_{0,i}(m_0))$. The rates $(R_0,R_1,R_2)$ are chosen so that the ensemble of generated sequences ``cover'' the set $\Ti_{\epsilon^{\prime}}^{(n)}(U_0, U_1, U_2|s_1^n,s_2^n)$ for all $s_1^n,s_2^n\in \Ti_{\epsilon^{\prime}}^{(n)}(S_1,S_2) $. Define the event
\begin{align*}
\mathcal{A}= & \{(S_1^n, S_2^n, U_0^n(m_0), U_1^n(S_1^n, m_0, m_1), \\
             &\qquad \qquad   U_2^n(S_2^n, m_0, m_2)) \not \in \Ti_{\epsilon^{\prime}}^{(n)}
          \text{ for all } m_0, m_1, m_2\}
\end{align*}
It can be shown using techniques similar to those used in the proof of the Lemma \ref{lem} that $\P(\mathcal{A}) \to 0$ as
$n \to \infty$ if
\begin{equation}
\label{eq:covb}
\setlength\arraycolsep{0.2em}
\begin{array}{rcl}
R_0 &>& I(U_0;S_1,S_2) + \delta(\e^{\prime}) , \\
R_0+R_1 &>& I(U_0;S_1,S_2) + I(S_2;U_1|S_1,U_0) + \delta(\e^{\prime}), \\
R_0+R_2 &>& I(U_0;S_1,S_2) + I(S_1;U_2|S_2,U_0) + \delta(\e^{\prime}),  \\
R_0+R_1+R_2 &>& I(U_0;S_1,S_2) +
I(S_2;U_1|S_1,U_0) \\
             & &\qquad +  I(S_1,U_1;U_2|S_2,U_0) + \delta(\e^{\prime})
\end{array}
\end{equation}
where $\delta(\e^{\prime}) \to 0$ as $\e^{\prime} \to 0$.

\emph{Encoding:} For each source sequence $(s_1^n,s_2^n)$, choose a
triple $(m_0,m_1,m_2) \in [1:2^{nR_0}] \times [1:2^{nR_1}] \times
[1:2^{nR_2}]$ such that $(s_1^n, s_2^n, u_0^n(m_0), u_1^n(s_1^n,m_0,m_1), u_2^n(s_2^n,m_0,m_2)) \in \Ti_{\epsilon^{\prime}}^{(n)}(S_1,S_2,U_0,U_1,U_2)$.
If there is no such triple, choose $(m_0,m_1,m_2) = (1,1,1)$. Then at time $i\in [1:n]$, the encoder transmits $x_i = x(s_{1i}, s_{2i}, u_{0i}(m_0), u_{1i}(s_1^n,m_0,m_1), u_{2i}(s_2^n,m_0,m_2))$.

\emph{Decoding:} Let $\e > 0$. Decoder 1 declares $\uh_1^n$ to be the estimate of the source $s_1^n$ if it is the unique sequence  such that $(\uh_1^n, u_0^n(m_0), u_1^n(\uh_1^n,m_0,m_1), y_1^n) \in \aep(S_1,U_0,U_1,Y_1)$ for some $(m_0,m_1) \in [1:2^{nR_0}] \times [1:2^{nR_1}]$. Similarly,
decoder 2 declares $\uh_2^n$ to be the estimate of the source $u_2^n$ if it is the unique sequence  such that $(\uh_2^n, u_0^n(m_0), u_2^n(\uh_2^n,m_0,m_2), y_2^n) \in \aep(S_2,U_0,U_2,Y_2)$ for some $(m_0,m_2) \in [1:2^{nR_0}] \times [1:2^{nR_2}]$.

\emph{Error events}: Assume that $(m_0,m_1,m_2)=(M_0,M_1,M_2)$ is selected for the source sequence  $(s_1^n, s_2^n)$. The error event for decoder 1 can be divided into two parts:
\begin{enumerate}
\item There exists a sequence $\hat{s}_1^n \neq S_1^n$ such that $(\hat{s}_1^n, u_0^n(M_0), u_1^n(\hat{s}_1^n,M_0,m_1), y_1^n) \in \aep(S_1,U_0,U_1,Y_1)$ for some $m_1 \in [1:2^{nR_1}]$. The probability of this event vanishes as  $n \to \infty$  if the following inequality is satisfied
\begin{align}
\label{dec1b}
H(S_1) + R_1 &< I(U_1,S_1; Y_1|U_0) + I(S_1;U_0)- \delta(\e).
\end{align}
  \item  There exists a sequence $\hat{s}_1^n \neq S_1^n$ and an $m_0 \neq M_0$ such that $(\hat{s}_1^n, u_0^n(m_0), u_1^n(\hat{s}_1^n,m_0,m_1), y_1^n)\in \aep(S_1,U_0,U_1,Y_1) $ for some $m_1 \in [1:2^{nR_1}]$. The probability of this event can be made arbitrarily small as $n \to \infty$ if
\begin{align}
H(S_1) + R_0 + R_1 &< I(U_0,U_1,S_1; Y_1) + \nonumber \\
& \qquad + I(U_0; S_1) - \delta(\e).
\label{dec2b}
\end{align}
\end{enumerate}
Similarly, the probability of error for decoder 2 can be made arbitrarily small as  $n \to \infty$ if
\begin{equation}
\label{dec3b}
\setlength\arraycolsep{0.2em}
  \begin{array}{rcl}
H(S_2) + R_2 &<& I(U_2,S_2; Y_2|U_0) + I(S_2;U_0)- \delta(\e), \\
H(S_2) + R_0 + R_2 &<& I(U_0,U_2,S_2; Y_2) + I(U_0; S_2)- \delta(\e).
\end{array}
\end{equation}
The rest of the proof follows by letting $\e^{\prime},\e \to 0$, then $n \to \infty$, and finally by eliminating $(R_0, R_1, R_2)$ from \eqref{eq:covb}, \eqref{dec1b}, \eqref{dec2b}, and \eqref{dec3b} using the
Fourier--Motzkin elimination algorithm.

\bibliographystyle{ieeetr} % try also ieeetr, apalike, acm, siam,

\end{document}